\documentclass[a4paper,12pt]{article} % DO NOT CHANGE THIS
\usepackage{graphicx} % DO NOT CHANGE THIS
\usepackage{caption} % DO NOT CHANGE THIS AND DO NOT ADD ANY OPTIONS TO IT

\usepackage{algorithm}
\usepackage{algorithmic}

\usepackage{newfloat}
\usepackage{listings}
\DeclareCaptionStyle{ruled}{labelfont=normalfont,labelsep=colon,strut=off} % DO NOT CHANGE THIS
\floatstyle{ruled}
\newfloat{listing}{tb}{lst}{}
\floatname{listing}{Listing}

\title{Complexity of Reasoning with Cardinality Minimality Conditions}
\author{
		Nadia Creignou,\textsuperscript{\rm 1}
		Fr\'ed\'eric Olive,\textsuperscript{\rm 1}
		Johannes Schmidt\textsuperscript{\rm 2}
\\    %Afiliations
{\footnotesize
    \textsuperscript{\rm 1}
		Aix Marseille Univ, CNRS, LIS, Marseille, France
}
		\\
{\footnotesize
    \textsuperscript{\rm 2}
		J\"onk\"oping University, Department of Computer Science and Informatics, School of Engineering, Sweden
}
		\\
{\footnotesize
		\{nadia.creignou, frederic.olive\}@lis-lab.fr, johannes.schmidt@ju.se
}
}

\usepackage{xspace}
\usepackage{color} %xcolor
\usepackage{amsbsy, amscd, amssymb, latexsym} 
\usepackage{amsthm}  
\usepackage[cmex10]{amsmath}
\usepackage{mathtools} % corrige quelques bugs de amsmath
\usepackage{tabularx}

\newcommand{\cloneFont}[1]{\mathsf{#1}}

\newtheorem{theorem}{Theorem}
\newtheorem{lemma}[theorem]{Lemma}

\newtheorem{proposition}[theorem]{Proposition}
\newtheorem{example}{Example}
\newtheorem{definition}{Definition}

\newlength\problemlength
\newcommand{\relevance}{{\sc Card-min}-{\sc Relevance}\xspace}
\newcommand{\outputname}[1]{
	{\ifthenelse{\equal{#1}{d}}{Question}{\ifthenelse{\equal{#1}{s}}{Output}{Output}}}
}% definition of a problem
\newcommand\problem[4]{%
\begin{list}{}{
\labelwidth\problemlength\labelsep.7em\rightmargin1.0em
\leftmargin\problemlength\advance\leftmargin by1.7em%2em
\parsep0ex\itemsep.2ex plus.1ex}
\item[#2 \hfill]
\item[\emph{Instance :} \hfill] #3
\item[\hspace{-.5ex}\emph{\outputname{#1}:} \hfill] #4
\end{list}
}
\newcommand\dproblem[3]{
	{\problem{d}{#1}{#2}{#3}}
}
\newcommand{\pap}{\mathcal{P}}
\newcommand{\tuple}[1]{\langle #1 \rangle}
\newcommand{\card}[1]{\left|#1\right|}
\newcommand{\s}{\mathcal{S}}
\newcommand{\overbar}[1]{\mkern 2.7mu\overline{\mkern-2.7mu#1\mkern-2.7mu}\mkern 2.7mu}
\newcommand{\IFF}{\textrm{ iff }}
\newcommand{\rxor}[1]{\rXOR{#1}}
\newcommand{\rxort}{\rXOR{3}}
\newcommand{\rnaet}{\rNAE{3}}
\newcommand{\cmskxort}{\CMS^*(\rxort)}
\newcommand{\cmsnaet}{\CMS(\rnaet)}
\newcommand{\DELTA}[1]{\Delta_{#1}^\mathrm{P}}
\newcommand{\reduc}{\leq_\text{m}^\text{P}}
\newcommand{\sat}{\textsc{Sat}\xspace}
\newcommand{\xor}{\ensuremath{\oplus}}

\newcommand{\Cardminsat}{\textsc{CardMinSat}\xspace}
\newcommand{\CMS}{\textsc{Cms}\xspace}
\newcommand{\THETA}[1]{\Theta_{#1}^\mathrm{P}}
\newcommand{\PP}{\mathrm{P}}
\newcommand{\NP}{\mathrm{NP}}
\newcommand{\var}{\operatorname{var}}
\newcommand{\T}{\protect\ensuremath{\mathrm{T}}}
\newcommand{\F}{\protect\ensuremath{\mathrm{F}}}
\newcommand{\Rddd}{\ensuremath{R_{\scriptscriptstyle 3\neq}^{\scriptscriptstyle 1/3}}}
\newcommand{\rOR}[1]{\ensuremath{\textsf{\small OR}^{#1}}}
\newcommand{\rXOR}[1]{\ensuremath{\textsf{\small XOR}^{#1}}}
\newcommand{\rNAE}[1]{\ensuremath{\textsf{\small NAE}^{#1}}}
\newcommand{\orn}[2]{\ensuremath{\mathrm{OR}^{#2}_{\scriptscriptstyle #1}}}
\newcommand{\nandn}[2]{\ensuremath{\mathrm{NAND}^{#2}_{\scriptscriptstyle #1}}}
\newcommand{\evenn}[2]{\ensuremath{\mathrm{EVEN}^{\scriptscriptstyle #1}_{\scriptscriptstyle #2}}}
\newcommand{\clos}[1]{\ensuremath{\langle#1\rangle}}

\newcommand{\closNex}[1]{\ensuremath{\langle#1\rangle_{\not\exists}}}
\newcommand{\closNN}[1]{\ensuremath{\langle#1\rangle_{\not\exists,\neq}}}

\newcommand{\cc}[3]{\ensuremath{\cloneFont{{#1}^{#2}_{#3}}}}
\newcommand{\cocl}{\protect\ensuremath{\cloneFont{C}}}
\newcommand{\daffine}{width-$2$ affine\xspace}

\begin{document}

\maketitle

\begin{abstract}
Many AI-related reasoning problems are based on the problem of satisfiability  of propositional formulas with some cardinality-minimality condition. While the complexity of the satisfiability problem (SAT) is well understood
when considering systematically all fragments of propositional logic within Schaefer’s framework (STOC 1978)
%
%  when restricting the structure of the formulas in a certain way,
%	
this is not the case when such minimality condition is added. We consider the \Cardminsat problem, which asks, given a formula $\phi$ and an atom $x$, whether $x$ is true in some cardinality-minimal model of $\phi$.  
We completely   classify  the computational complexity of the \Cardminsat problem within Schaefer's framework, thus paving the way for a better understanding of the tractability frontier of many AI-related reasoning problems. To this end we use advanced algebraic tools developed by (Schnoor\;\&\;Schnoor 2008) and (Lagerkvist 2014).
\end{abstract}

\section{Introduction}

In many AI-related reasoning problems some notion of minimality is involved. Typically in belief change, e.g. revision or update, one of the basic principles is the principle of minimal change. We want to revise/update  an agent's belief set  by some new information. To this end we retain only those models of new information that have minimal distance to  the models of the original agent's belief set. 
In the belief revision context distance between models is defined by the symmetric set difference of the atoms assigned to true in the compared models, and Dalal's operator \cite{Dalal88},  for instance, seeks to minimize the cardinality of this set. In abduction we search for an explanation (a set of literals) that is consistent with a given theory and which, together with this theory, logically entails all manifestations. It is natural  to be interested not in all explanations but only in the minimal ones. Different notions of minimality might be considered, in particular minimality w.r.t set inclusion or w.r.t. cardinality \cite{jacm/EiterG95}. 

\smallskip

In this paper, we focus on cardinality-minimality. With such a minimality condition the related reasoning tasks often give rise to $\THETA 2$-complete problems (the class $\THETA 2$ is located at the second level of the polynomial hierarchy: polynomial time  with only a logarithmic number of calls to the $\NP$-oracle). For instance, model checking and implication are  $\THETA 2$-complete for Dalal's revision operator \cite{ai/EiterG92, jcss/LiberatoreS01}. The relevance problem for abduction  with a cardinality-minimality condition, deciding whether a literal belongs to a cardinality-minimal explanation, is  $\THETA 2$-complete when dealing with Horn formulas \cite{jacm/EiterG95}. 

\smallskip

Propositional formulas play an important role in AI-reasoning problems. Since most relevant problems are \mbox{intractable} in full propositional logic, it is a natural question whether syntactic restrictions on the involved formulas can lead to tractable problems. Schaefer's framework offers an ideal framework to investigate this issue. It considers formulas in generalized conjunctive normal form and allows to systematically consider all fragments of propositional logic. Indeed, Schaefer's famous theorem \cite{stoc/Schaefer78} shows that the SAT problem becomes tractable under some syntactic restrictions such as Horn, dual Horn, Krom or affine formulas, and remains intractable in all other, nontrivial, cases. 
Since then Schaefer's approach has been taken on numerous problems, among others on circumscription, abduction, and argumentation problems \cite{Nordh04, ai/NordhZ08, CreignouE014}. Tools from universal algebra prove to be a valuable tool for such endeavors, in particular when the problem questions are stable under introduction of existentially quantified variables and equality constraints \cite{dagstuhl/CreignouV08}. Unfortunately, cardinality is not preserved under such introduction. Therefore, in this paper, we resort to advanced algebraic tools built around the concept of a \emph{weak base} \cite{scs08,Lagerkvist14}.

\smallskip

There is a prototypical satisfiability problem for the class $\THETA 2$, that could enlighten the complexity of many reasoning problems involving a cardinality-minimality condition: It is the \Cardminsat problem, 
which asks, given a formula $\phi$ and an atom $x$, whether $x$ is true in some cardinality-minimal model of $\phi$. It  provides a standard hard problem that can be useful to prove hardness results, especially in the context of knowledge representation and belief change. For instance, in \cite{CreignouPW18} the relevance problem for abduction  mentioned above was proved to be $\THETA 2$-complete for the combined Horn-Krom case. The $\THETA 2$-hardness reduction used in \cite{CreignouPW18} is much easier than the one previously obtained in \cite{jacm/EiterG95} for the Horn case, because it starts from the more closely related problem \Cardminsat, restricted to conjunctions of positive 2-clauses. Similarly, the model checking and implication problems associated with Dalal's operator were proved in~\cite{CreignouPW18} to be $\THETA 2$-complete for the combined Horn-Krom case by a reduction from the \Cardminsat problem. Also the model checking problem  associated with a syntactic  revision operator  for belief bases using a cardinality-maximality criterion was proved to be  $\THETA 2$-complete for the combined Horn-Krom case in this way \cite{CreignouKP17}.
 
\smallskip
 
Our main contribution is a complete complexity classification of the \Cardminsat problem in Schaefer's framework, which   opens the door for a better understanding of the complexity of many reasoning problems. 
As an illustration we prove that the above mentioned relevance problem for abduction remains  $\THETA 2$-complete when restricted to affine formulas (conjunctions of XOR-clauses).

\section{Preliminaries}\label{sec:preliminaries}

\paragraph{Propositional logic.} 

We assume familiarity with propositional logic. 
A \emph{literal} is a variable (or an atom) $x$ (positive literal) or its negation $\neg x$  (negative literal). 
A \emph{clause} is a disjunction of literals. For any integer $k\ge 1$,  a $k$-clause is a clause  containing at most $k$ literals.  An \emph{XOR-clause} is  a clause in which the usual connective ``or'' is replaced by the exclusive-or connective, denoted by $\oplus$.
A CNF-formula (resp., an XOR-CNF-formula) is a conjunction of clauses (resp., XOR-clauses), a $k$-CNF-formula is a conjunction of $k$-clauses.
For space economy we use occasionally the shorthands $\overline{x} := \neg x$ and $xy := x \wedge y$.
Given a formula $\phi$, we denote by $\var(\phi)$ the variables of $\phi$. 
% not needed?
%Analogously, for a set of formulas $F$, $\var(F)$ denotes $\bigcup_{\phi \in F}\var(\phi)$. 
% We identify finite $F$ with the conjunction of all formulas from $F$, that is, $\bigwedge_{\phi \in F} \phi$. 
 A mapping $\sigma\colon \var(\phi) \mapsto \{0,1\}$ is called an \emph{assignment} to the variables of $\phi$.
An assignment $\sigma$ satisfies a (XOR-)CNF-formula $\phi$ if $\sigma$ satisfies all (XOR-)clauses simultaneously. In this case $\sigma$ is called a \emph{model} of $\phi$. 
We call a variable $x \in \var(\phi)$ \emph{frozen} if $x$ is assigned the same value in all models of $\phi$.
The \emph{weight} or \emph{cardinality} of an assignment $\sigma$, denoted by $\vert\sigma\vert$,  is the number of variables $x$ such that $\sigma(x)=1$. A \emph{cardinality minimal} model of $\phi$ is a model of $\phi$ of minimum cardinality among all models of $\phi$.
For two formulas $\psi, \phi$ we write $\psi  \models \phi$ if every model of $\psi$ also satisfies $\phi$. The two formulas are equivalent, $\psi  \equiv \phi$, if they have the same set of variables and the same set of models.
 Observe that any XOR-clause is equivalent to a linear equation over the two-elements field, of the form $x_1\oplus \ldots \oplus x_n=a$ where $a\in \{0,1\}$.

\paragraph{Schaefer's framework.}

% 
% a little outdated
% For a deeper introduction into Schaefer's CSP framework, consider the article of Böhler~et~al.~\cite{bcrv04}.

A \emph{Boolean relation} of arity $k\in\mathbb N$ is a relation $R\subseteq\{0,1\}^k$, and a \emph{constraint} $C$ is a formula $C = R(x_1,\dots,x_k)$, where $R$ is a $k$-ary Boolean relation, and $x_1,\dots,x_k$ are (not necessarily distinct) variables.
%
% not needed?
% If $V$ is a set of variables and $u$ a variable, then $C[V/u]$ denotes the constraint obtained from $C$ by replacing every occurrence of every variable of $V$ by $u$.
% \nadia{it seems that this notation $C[V/u]$ is never used in the paper}
% %
An assignment $\sigma$ \emph{satisfies $C$}, if $(\sigma(x_1),\dots,\sigma(x_k))\in R$.
A \emph{constraint language} 
$\Gamma$ is a finite set of Boolean relations, and a \emph{$\Gamma$-formula} is a conjunction of constraints using relations from $\Gamma$. Note that we do not consider infinite constraint languages in this paper.
Finally, a $\Gamma$-formula $\phi$ is \emph{satisfied} by an assignment $\sigma$, if $\sigma$ simultaneously satisfies all constraints in it.
In such a case $\sigma$ is also called a \emph{model of $\phi$}.
%
%not needed?
% Whenever a $\Gamma$-formula or a constraint is logically equivalent to a single clause or term or literal, we treat it as such. \fred{What does-it mean?}
%
We say that a $k$-ary relation $R$ is \emph{defined  by} a formula $\phi$  if $\phi$ is a formula over $k$ distinct variables $x_1,\ldots,x_k$ and $\phi\equiv R(x_1,\ldots, x_k)$.

Moreover, we say that a Boolean relation  $R$ is: 
\begin{itemize}
\item \emph{Horn} (resp., \emph{dual-Horn}) if it is definable by a CNF-formula  $\phi$  that contains at most one positive (resp., negative) literal per each clause,
\item \emph{Krom } if it is definable by a $2$-CNF-formula,
\item \emph{affine} it is definable by an XOR-CNF formula, or equivalently by  a formula $\phi$ that is a conjunction of linear equations of the form $x_1\oplus \ldots \oplus x_n=a$, where $a\in \{0,1\}$,

\item \emph{width-2-affine} it is definable by an XOR-$2$-CNF formula, or equivalently by  a formula $\phi$ that is a conjunction of linear equations involving each at most two variables, that is either of the form $x_1=a$ or of the form $x_1\oplus  x_2=a$, where $a\in \{0,1\}$.
% \item \emph{Essentially negative} if every clause in $\phi$ is either negative or unit positive. $R$ is \emph{essentially positive} if every clause in $\phi$ is either positive or unit negative.
\item \emph{$1$-valid} (resp., \emph{$0$-valid}) if $(1,\ldots, 1)\in R$ (resp., $(0,\ldots, 0)\in R$).

\item \emph{complementive} if for every tuple $(t_1, \dots, t_k) \in R$ also 
$(1-t_1, \dots, 1-t_k) \in R$. 
\end{itemize}
Furthermore, we say a relation is \emph{Schaefer} if it is Horn, dual-Horn, Krom, or affine.
Finally, for a property $\cal{P}$ of a relation, we say that a constraint language $\Gamma$ is ${\cal P}$ if all relations in $\Gamma$ are ${\cal P}$.

\smallskip

\noindent
We define the unary relations $\T = \{1\}$, $\F = \{0\}$, and the 6-ary relation \\
$\Rddd=\{100011, 010101, 001110\}$.
We denote by $\rOR{k}$ the $k$-ary OR, by $\nandn{}{k}$ the $k$-ary NAND, and by $\rXOR{k}$ the $k$-ary XOR. The relation $\evenn{k}{}$ contains all $k$-ary tuples which contain an even number of 1's. The relation $\evenn{k}{k \neq}$ denotes the $2k$-ary relation defined by $\evenn{k}{}(x_1, \dots, x_k) \land (x_1 \neq x_{k+1}) \land \dots \land (x_k \neq x_{2k})$.

\bigskip

In the following definition we introduce different notions of closure for a constraint language. 

\begin{definition}
	\begin{enumerate}
	\item The set $\clos{\Gamma}$ is the smallest set of relations that contains $\Gamma$, the equality constraint, $=$, and which is closed under primitive positive first order definitions, that is, if $\phi$ is a $\Gamma \cup \{=\}$-formula and $R(x_1, \dots, x_n) \equiv \exists y_1 \dots \exists y_l \phi(x_1, \dots, x_n,y_1, \dots, y_l)$, then $R \in \clos{\Gamma}$. In other words, $\clos{\Gamma}$ is the set of relations that can be expressed as a $\Gamma \cup \{=\}$-formula with existentially quantified variables.
	\item The set $\closNex{\Gamma}$ is the set of relations that can be expressed as a $\Gamma \cup \{=\}$-formula (no existentially quantified variables are allowed).
	\item The set $\closNN{\Gamma}$ is the set of relations that can be expressed as a $\Gamma$-formula  (neither the equality relation nor existentially quantified variables are allowed).
	\end{enumerate}
\end{definition}

%\begin{example}
%Let $\Gamma := \{x_1 \rightarrow x_2\}$ and $R(x_1,x_2) := (x_1 = x_2)$. We can express $R$ as $\Gamma$-formula via $R(x_1, x_2) \equiv (x_1 \rightarrow x_2) \land (x_2 \rightarrow x_1)$. Thus, $R \in \closNN{\Gamma}$.
%\end{example}
%
\begin{example}\label{ex:example1}
Let $\Gamma = \{R\}$, $R(x_1, x_2) = (x_1 \rightarrow x_2)$, and $S(x_1,x_2) = (x_1 = x_2)$. We can express $S$ as $\Gamma$-formula via $S(x, y) \equiv R(x, y) \land R(y, x)$. Thus, $S \in \closNN{\Gamma}$.
\end{example}
The set $\clos{\Gamma}$ is called a \emph{relational clone} or a \emph{co-clone} with \emph{base} $\Gamma$ \cite{ipl/BohlerRSV05}. 
Notice that for a co-clone $\cocl$ 
% \fred{n'est-ce pas maladroit de prendre la même lettre que pour désigner une contrainte?}
and a constraint language  $\Gamma$ the statements $\Gamma \subseteq \cocl$, $\clos{\Gamma} \subseteq \cocl$, $\closNex{\Gamma} \subseteq \cocl$, and $\closNN{\Gamma} \subseteq \cocl$ are equivalent. 
Throughout the paper, we refer to different types of Boolean relations and corresponding co-clones following Schaefer's terminology \cite{stoc/Schaefer78}.
Table~\ref{table:weak_bases} gives a complete list of all finitely generated co-clones with minimal weak base from Lagerkvist \cite{Lagerkvist14}. For clause type descriptions and simpler bases (not minimal weak bases) we refer to \cite{ai/NordhZ08} and \cite{ipl/BohlerRSV05}, respectively.
%\ref{table:weak_bases}.

A graph representation of the co-clone structure can be found in figure~\ref{fig:posts-lattice}.
This graph is usually called Post's lattice~\cite{pos41}.
Some important properties/names are labeled besides the respective co-clone.
Informally explained, every vertex corresponds to a co-clone while the edges model the containment relation in this lattice structure.

\paragraph{Complexity Classes.} 

All complexity results in this paper refer to classes in the Polynomial Hierarchy (PH) \cite{Papadimitriou94}. The building blocks of PH are the classes  $\PP$ and $\NP$ of decision problems solvable in deterministic, resp.\ non-deterministic, polynomial time. The class  $\DELTA 2$ is   the class of decision problems that can be decided by a  deterministic Turing machine  in polynomial time using an oracle for the class $\NP$. One can put restrictions on the number of oracle  calls.  If on input $x$ with
$|x| = n$ at most $O(\log n)$ calls 
to the $\NP$ oracles are allowed, then we get the class 
$\PP^{\NP[O(\log n)]}$, which is also 
referred to as  $\THETA 2$. 
A large collection of $\THETA 2$-complete problems can be obtained from 
\cite{jcss/Krentel88,mst/GasarchKR95}. For the reductions we employ polynomial many-one reductions, denoted by $\reduc$.

\section{\Cardminsat}
\label{sec:cardminsat}

We aim  at studying the following natural variant of SAT and  analyzing its complexity.

\dproblem{\Cardminsat}%
{A propositional formula $\phi$ and an atom $x$.}%
{Is $x$ true in a cardinality-minimal model of $\phi$?}

% \Cardminsat \\
% {\em Instance:} Propositional formula $\phi$ and  atom $x$.  \\
% {\em Question:} Is $x$ true in a cardinality-minimal model of $\phi$?

This problem is one of the prototypical problems of the class $\THETA 2$, see \cite{icalp/Wagner88,CreignouPW18}. 
It makes sense to study whether syntactic restrictions on the formulas make the problem easier and to go through a more fine-grained complexity study of $\Cardminsat$, in the following also denoted $\CMS$. To this aim
we
propose to investigate this problem within Schaefer's framework.  Hence we consider the following problem, in which $\Gamma$ is a constraint language, i.e., a finite set of Boolean relations. 

\dproblem{$\CMS(\Gamma)$}%
{A $\Gamma$-formula $\phi$ and an atom $x$.}%
{Is $x$ true in a cardinality-minimal model of $\phi$?}

%$\CMS(\Gamma)$\\
%{\em Instance:} A $\Gamma$-formula $\phi$ and  atom $x$. \\
%{\em Question:} Is $x$ true in a cardinality-minimal model of $\phi$?

Analogously we denote by $\sat(\Gamma)$ the Boolean satisfiability problem for $\Gamma$-formulas.
Our goal is to obtain a complete complexity classification of $\CMS(\Gamma)$, depending on~$\Gamma$. This issue has already been settled in the literature within the Krom fragment. 

% theorem
% theorem
\begin{theorem}\label{theorem:krom_classification}{\rm \cite{CreignouPW18}}
Let $\Gamma$ be a  Krom constraint language.
 If $\Gamma$ is \daffine or Horn, then  $\CMS(\Gamma)$   is decidable in
polynomial time. Otherwise it is $\THETA{2}$-complete. 
\end{theorem}

We extend this result and obtain a complete complexity classification in all fragments of propositional logic.

% theorem
% theorem
\begin{theorem}\label{theorem:classification}
Let $\Gamma$ be a  constraint language.
 If $\Gamma$ is \daffine or Horn or 0-valid, then  $\CMS(\Gamma)$   is decidable in
polynomial time. Otherwise it is $\THETA{2}$-complete. 
\end{theorem}

Note that $\CMS(\Gamma)$ is trivial for 0-valid formulas (the answer is always "no").
The complexity classification of $\CMS$ in the Krom fragment had been obtained by means of partial frozen co-clones. While these partial frozen  co-clones are well described within the Krom fragment \cite{NordhZ09}, they are  only partially known in the full range of propositional logic. For this reason in order to get the complete classification we use   another set of tools. In particular we will use a restricted  notion of closure, and build on  the notion of weak bases introduced in \cite{scs08}. This is described in the next section.

\section{Technical tools} \label{sec:tec-tools}

\subsection{Proof's method}\label{sec:method}

The above introduced closure operator $\clos{\cdot}$ on sets of Boolean relations is relevant in order to obtain complexity results for the satisfiability problem. Indeed, assume that $\Gamma_1\subseteq\clos{\Gamma_2}$. Then a $\Gamma_1$-formula can be transformed into a satisfiability-equivalent $\Gamma_2$-formula, thus showing that $\sat(\Gamma_1)$ can be reduced in polynomial time to $\sat(\Gamma_2)$ \cite{jeavons98}. Hence, the complexity of  $\sat(\Gamma)$ depends only on the co-clone $\clos\Gamma$. 
Accordingly, in order to obtain a full
complexity classification for the satisfiability problem one only
has to study the co-clones.

Unfortunately, since we are here interested in cardinality-minimal models, we cannot a priori only study the co-clones.
Indeed, existential variables and equality constraints that may occur when transforming  a $\Gamma_1$-formula into a satisfiability-equivalent $\Gamma_2$-formula are problematic, as they can change the set of models and the cardinality of each model. 
Therefore, we will use a more restricted notion of closure, namely the above introduced closure operator $\closNN{\cdot}$. This operator avoids existential quantifiers and equality constraints. The only operation to express relations in $\closNN{\Gamma}$ is conjunction of $\Gamma$-constraints (see e.g. example~\ref{ex:example1}). Consequently, when replacing in a reduction a relation $R \in \closNN{\Gamma}$ by its representing $\Gamma$-formula, $R$ is represented \emph{exactly}: no new variables are introduced and no constraints other than those built on $\Gamma$ are allowed (in particular no equality constraints). Therefore, any reduction based on the closure operator $\closNN{\cdot}$ preserves exactly the set of models, and, a fortiori, all cardinality-minimal models. Hence, we obtain the following property.

% proposition
\begin{proposition}\label{thm:R in NNclosure(G) => cms(R)<cms(G)}
Let $\Gamma$ be  a  constraint language and $R$ be a relation. 
 
	 If $R\in\closNN{\Gamma}$ then $\CMS(R)\reduc\CMS(\Gamma)$.
\end{proposition}

\noindent
The proof of  our complete classification will consist in a systematic exploration of the co-clones lattice, yet reductions can only be obtained via the restrictive operator $\closNN{\cdot}$, not via the more expressive, co-clone generating, operator $\clos{\cdot}$. In this context, the concept of a \emph{weak base} is important \cite{scs08}. A weak base $B$ for a co-clone $\cocl$ has the property that 
(1) $\clos{B}=\cocl$, and (2) $B\in\closNex{\Gamma}$ for any $\Gamma$ such that
$\clos{\Gamma}=\cocl$. The existence of a weak base for each co-clone has been shown by Schnoor and Schnoor \cite{scs08}. For a finitely generated co-clone $\cocl$ there even exists a single relation weak base. If such a weak base $B$ is in addition \emph{irredundant} (that is, the matrix representation does not contain redundant columns), it holds even that $B\in\closNN{\Gamma}$ for any $\Gamma$ such that $\clos{\Gamma}=\cocl$. Lagerkvist (2014) has identified \emph{minimal} weak bases for all finitely generated co-clones \cite{Lagerkvist14}. A relation $R$ is minimal, if (1) $R$ is irredundant, (2) $R$ contains no fictitious coordinates, (3) there is no $R' \subsetneq R$, such that $\clos{R} = \clos{R'}$. A coordinate $i$ is called \emph{fictitious} if its value has no influence on the membership of a tuple, that is,  $(x_1, \dots, x_{i-1}, 0, x_{i+1}, \dots, x_k) \in R$
if and only if $(x_1, \dots, x_{i-1}, 1, x_{i+1}, \dots, x_k) \in R$. Table~\ref{table:weak_bases} contains a complete list of all finitely generated co-clones with minimal weak bases.

\medskip

The proof method to obtain our complete classification will use the minimal weak bases as follows. In order to show a hardness result for all constraint languages generating a certain co-clone $\cocl$, we pick a minimal weak base $B$ of $\cocl$ and show that $\CMS(B)$ is hard. This implies then hardness of $\CMS(\Gamma)$ for any $\Gamma$ such that $\clos{\Gamma} = \cocl$ by applying Proposition~\ref{thm:R in NNclosure(G) => cms(R)<cms(G)} (because $B$ is a minimal weak base, it is irredundant, and we hence have that $B \in \closNN{\Gamma}$). We state this in the following proposition.

% proposition
\begin{proposition}\label{thm:weak_base_method}
Let $\cocl$ be a co-clone and $B$ be a minimal weak base of $\cocl$.
Then it holds that $\CMS(B)\reduc\CMS(\Gamma)$ for any $\Gamma$ such that
$\clos{\Gamma} = \cocl$.
\end{proposition}

\smallskip

To start with, we need hardness results for some specific relations, and this is the aim of the next section.

% subsection
\subsection{Specific hardness results}

We give here some hardness results for some specific relations, they will be used in order to get hardness results for co-clones in the next section. The classification obtained in Theorem \ref{theorem:krom_classification} for the Krom fragments implies the following result. 

% lemma
\begin{lemma}\label{lemma:cardminsat_completeness_or}
$\CMS(\rOR{2})$ is  $\THETA 2$-hard.
\end{lemma}

The next result will also be a cornerstone in our classification proof. 

% lemma
\begin{lemma}\label{lemma:cardminsat_completeness_xor}
$\CMS(\rXOR{3})$ is $\THETA{2}$-hard.
\end{lemma}

\begin{proof}[Proof sketch]
Recall that $\rXOR{3}(x,y,z) \equiv (x \xor y \xor z)$ and  $\rXOR{4}(x,y,z,u) \equiv (x \xor y \xor z \xor u)$. 
Here we will also use the ternary relation $\rNAE{3} =\{0,1\}^3\setminus \{000,111\}$ and the problem $\CMS^*(\Gamma)$, defined as follows:

\dproblem{$\CMS^*(\Gamma)$}
{A $\Gamma$-formula $\phi$, atom $x$, integer $k$.} 
{Is $x$ true in a cardinality-minimal model of $\phi$ and is this cardinality $\leq k$?}

\noindent
The proof consists of the following sequence of reductions.

\begin{enumerate}
\item $\CMS(\rOR{2})\reduc\CMS(\rNAE{3})$
\item $\cmsnaet\reduc\cmskxort$
\item $\CMS^*(\rXOR{3})\reduc \CMS(\rXOR{4})$
\item $\CMS(\rXOR{4}) \reduc \CMS(\rXOR{3},\,\rXOR{2})$
\item $\CMS(\rXOR{3},\,\rXOR{2}) \reduc \CMS(\rXOR{3})$
\end{enumerate}

\noindent
Then the result follows from Lemma \ref{lemma:cardminsat_completeness_or}.

We now give the reductions.

% 1. OR2 < NAE3
\paragraph{\normalfont 1. $\CMS(\rOR{2})\reduc\CMS(\rNAE{3})$.}\ \\

To each constraint $\rOR{2}(x,y)$ we associate the constraint $\rNAE{3}(x,y,f)$ where $f$ is a   fresh variable.
Observe that $\rOR{2}(x,y)\equiv \rNAE{3}(x,y,0)$.  Therefore
the idea is to use $f$   in the place of 0  as a  global variable  (that is the same for all constraints) and to force it to take value 0 in all cardinality-minimal models. This can be done by giving a weight $N$ to $f$ big enough.  For this we add the constraint $\rNAE{3}(f,f,t)$, which expresses $f\ne t$,  and   $N$   constraints $\rNAE{3}(f_j,f_j,t)$, where the $f_j$, for $j=1,\ldots , N,$ are fresh variables. This ensures that if $f=1$ then $f_1,\ldots, f_N=1$.  Observe moreover that since $\rNAE{3}$ is a complementive relation, the built $\rNAE{3}$-formula is satisfiable if and only if it has a model with $f=0$. Taking $N>n$ where $n$ is the number of variables of the original formula ensures that $f=0$ in any cardinality-minimal model.

\newcommand{\sg}{\sigma}
\newcommand{\plus}[1]{{#1}^{+}}
\newcommand{\moins}[1]{{#1}^{-}}
\newcommand{\et}{\wedge}
\newcommand{\Et}{\bigwedge}

% 2. NAE3 < XOR3
\paragraph{\normalfont 2. $\cmsnaet\reduc\cmskxort$.} \ \\

\smallskip

\noindent
Given a $\rnaet$-formula \ $\phi = \bigwedge_{i=1}^{m} \rnaet(x_i,y_i,z_i)$ \  and an atom $x\in\var(\phi)$, we want to construct an $\textsf{\small XOR}_3$-formula $\phi'$, an atom $x'$ and an integer $k$ such that: 

% equation
\begin{equation}\label{eq:reduc nae xor3}
\left(\begin{array}{c}
\text{$x$ belongs to a} \\ \text{minimal model of $\phi$}
\end{array}\right)
\text{ iff  }
\left(\begin{array}{c}
\text{$x'$ belongs to a minimal model} \\ \text{of $\phi'$ that has cardinality $\leq k$.} 
\end{array}\right)
\end{equation}

\noindent
To build the triple $(\phi',x',k)$, we set $x'=x$ and we choose $\phi'$ as the $\rxor3$-formula: 
\[
\phi' = \rxort(t,t,t) \ \et \ 
\Et_{i=1}^{m} 
\left\{\begin{array}{l}
\rxor3 (x_i,y_i,\alpha_i)	\et {\displaystyle \Et_{j=2}^{N}} \rxor3 (\alpha_i,\alpha_i^j,t) 	\quad \et \\ 
\rxor3 (x_i,z_i,\beta_i)	\et {\displaystyle \Et_{j=2}^{N}} \rxor3 (\beta_i,\beta_i^j,t) 		\quad \et \\ 
\rxor3 (y_i,z_i,\gamma_i)	\et {\displaystyle \Et_{j=2}^{N}} \rxor3 (\gamma_i,\gamma_i^j,t)
\end{array}
\right\}.
\]
Here, the atoms not in $\var(\phi)$ 
--- namely,  $t$ and the variables of the form $\alpha_i$, $\beta_i$, $\gamma_i$, $\alpha_i^j$, $\beta_i^j$ and $\gamma_i^j$ ---, 
are fresh variables. 
For the integer $N$, we take the cardinality of $\var(\phi)$. This will ensure that the cardinality of any model of $\phi$ is strictly smaller than~$N$. (Notice that such a model  maps at least one variable to $0$). 
At last, we  set $k=(m+1)N$ for reasons that will become clarified in the sequel of the proof. 
\smallskip

Notice that  any model $\tau$ of $\phi'$ maps $t$ to $1$, which entails  
\begin{enumerate}
\item  $\tau(\alpha_i^j)=\tau(\alpha_i)$, 
\item  $\tau(\beta_i^j)=\tau(\beta_i)$ and 
\item  $\tau(\gamma_i^j)=\tau(\gamma_i)$ 
\end{enumerate}
for all $i,j$. Thus, for each $i$ the variables $\alpha_i^j$ are used to boost the weight of $\alpha_i$: mapping $\alpha_i$ to $1$ in an assignment automatically increases the cardinality of this assignment by $N$ (and similarly for $\beta_i$ and $\gamma_i$).

\smallskip

Besides, observe that any assignment $\sg$ on $\var(\phi)$ can be extended in a unique way to a model of $\phi'$ by carefully choosing the values of the fresh variables. We denote by~$\plus\sg$ this extension. 
In the particular case where $\sg$ \emph{satisfies} $\phi$, we obtain:

% equation
\begin{equation}\label{EQ:de sg à sg+}
\card{\plus\sg} = \card{\sg} + mN.
\end{equation}

\smallskip

Conversely, if $\tau$ is a model of $\phi'$, then for any $i$, the total weight of $\alpha_i, \beta_i, \gamma_i$ in $\tau$ is $1$ if $\tau\models\rnaet(x_i,y_i,z_i)$ and $3$ otherwise. 
Denoting by $\moins\tau$ the restriction of $\tau$ to $\var(\phi)$, one can write
\, $\card{\tau} = \card{\moins\tau} + (m+2k)N$, \, 
where $k$ is the number of $i$ for which $\moins\tau\not\models\rnaet(x_i,y_i,z_i)$. It yields, for any assignment $\tau$ \emph{which satisfies} $\phi'$: 

% equation
\begin{equation}\label{EQ:de tau à tau-}
\moins\tau\models\phi \IFF \card{\tau} = \card{\moins\tau} + mN.
\end{equation}
This can be rephrased, since  $\card{\tau}$ has the shape $\card{\moins\tau} + (m+2k)N$ and since $\card{\moins\tau}<N$: 

% equation
\begin{equation}\label{EQ:de tau à tau- bis}
\moins\tau\models\phi \IFF \card{\tau} < (m+1)N.
\end{equation}
Now we can prove that Equivalence~\eqref{eq:reduc nae xor3} holds when taking $k=(m+1)N$.  

% =>
\framebox{$\Rightarrow$}
Let $\sg$ be a minimal model of $\phi$ containing $x$. Then $x$ also belongs to $\plus\sg$ which is a model of $\phi'$. Besides $\plus\sg$ is cardinality-minimal among the models of $\phi'$. 
For assume, towards a contradiction, that some model $\tau$ of $\phi'$ fulfils $\card{\tau}<\card{\plus\sg}$. 
Since $\sg\models\phi$, Equation~\eqref{EQ:de sg à sg+} ensures that $\card{\plus\sg} = \card{\sg} + mN < (m+1)N$. Therefore, $\card{\tau} < (m+1)N$ and hence by~\eqref{EQ:de tau à tau- bis}: $\moins\tau\models\phi$. 
It then entails that $\card{\tau} = \card{\moins\tau} + mN$ by~\eqref{EQ:de tau à tau-} and the inequality $\card{\tau}<\card{\plus\sg}$ implies $\card{\moins\tau} + mN<\card{\sg} + mN$, that is $\card{\moins\tau} < \card{\sg}$. 
Finally, $\moins\tau$ is a model of $\phi$ of cardinality strictly smaller than $\card\sg$: a contradiction. 

% <=
\framebox{$\Leftarrow$}
Let $\tau$ be a minimal model of $\phi'$ containing $x$ and fulfilling $\card\tau < (m+1)N$. Then $x$ also belongs to $\moins\tau$ which is a model of $\phi$, by~\eqref{EQ:de tau à tau- bis}. Besides $\moins\tau$ is cardinality-minimal among the models of $\phi$. 
Otherwise it exists a model $\sg$ of $\phi$ such that $\card{\sg}<\card{\moins\tau}$. But we have  
in the one hand, $\card{\moins\tau} = \card\tau - mN$ by~\eqref{EQ:de tau à tau- bis} 
and in the other hand, $\card{\sg} = \card{\plus\sg} - mN$ by~\eqref{EQ:de sg à sg+}. 
Hence inequality $\card{\sg}<\card{\moins\tau}$ yields $\card{\plus\sg} < \card{\tau}$. 
As $\plus\sg\models\phi'$, we finally get a model of $\phi'$ of cardinality strictly smaller than $\card\tau$: a contradiction. 
\qed

\bigskip

% 3. XOR3 < XOR4
\paragraph{\normalfont 3. $\CMS^*(\rXOR{3})\reduc \CMS(\rXOR{4})$.}\ \\

\noindent
Let $(\phi, x, k)$ be an instance of $\CMS^*(\rXOR{3})$, where
\[
\phi = \Et_{i=1}^{p} \rXOR{3}(x_{i},y_{i}, z_{i}).
\]

\noindent
We consider $k$ fresh variables $\alpha_1,\dots, \alpha_{k}$ and we set:
\[
\phi' = \Et_{i=1}^{p}\Et_{j=1}^{k} \rXOR{4}(x_{i},y_{i}, z_{i},\alpha_j).
\]

Observe that in all models of $\phi'$, equalities $\alpha_1=\dots =\alpha_k$ hold. Moreover $\phi'$ has one model of weight $k$, which assigns the $\alpha_j$'s to~1 and all other variables to~0. 
Each model $\sg$ of $\phi$  can be extended into a model of $\phi'$ by mapping each $\alpha_j$ to $0$. We denote by $\plus\sg$ this extension. Notice that in this case, $\card{\sg}=\card{\plus\sg}$.
Conversely, for any model $\tau$ of $\phi'$, we denote by $\moins\tau$ its restriction to $\var(\phi)$. It is easily seen that if $\tau(\alpha_j)=0$ for all $j$'s then $\moins\tau\models\phi$ and in this case, $\card{\tau}=\card{\moins\tau}$.

\smallskip

Let's now prove that the function $(\phi, x, k)\mapsto (\phi',x)$ is a reduction from $\CMS^*(\rXOR{3})$ to $\CMS(\rXOR{4})$.

\smallskip

Assume that $(\phi,x,k)$ is a positive instance of $\CMS^*(\rXOR{3})$ and call $\sg$ a minimal model of $\phi$ containing $x$ and such that $\card{\sg}\leq k$. Then $\plus\sg\models\phi'$ and $\plus\sg(x)=1$. Furthermore, $\plus\sg$ is minimal among the models of $\phi'$ since for any model $\tau$ of $\phi'$: If $\tau(\alpha_j)=1$ for all $j$'s, then $\card{\tau}\geq k\geq  \card{\sg}=\card{\plus\sg}$. If $\tau(\alpha_j)=0$ for all $j$'s, then $\moins\tau$ is a model of $\phi$, which satisfies  $\card{\moins\tau}\geq \card{\sg}$ by minimality of $\sg$. Since in this case $\card{\moins\tau}=\card{\tau}$  we get $\card{\tau}\geq \card{\sg}=\card{\plus\sg}$. Thus, $\plus\sg$ is a minimal model of $\phi'$ containing $x$ and hence, $(\phi',x)$ is a positive instance of $\CMS(\rXOR{4})$.

\smallskip

Conversely, assume $(\phi',x)$ is a positive instance of $\CMS(\rXOR{4})$ and call $\tau$ a minimal model of $\phi'$ containing $x$. Then $\card{\tau}\leq k$ since $\tau$ is minimal and since the assignment that maps every $\alpha_j$ on $1$ and all other variables on $0$ is a model of $\phi'$ of cardinality~$k$.
From $\tau(x)=1$ we get $\tau(\alpha_j)=0$ for all $j$' (otherwise $\card{\tau}$ would be greater than or equal to $k+1$). Hence, $\moins\tau$ is a model of $\phi$ containing $x$ such that $\card{\moins\tau}\le k$. It remains to verify that $\moins\tau$ is minimal among the models of $\phi$. Consider any model $\sg$ of $\phi$, then $\plus\sg$ is a model of $\phi'$. By minimality of $\tau$, $\card{\plus\sg}\geq\card{\tau}$, that is, $\card{\sg}\geq\card{\moins\tau}$.
Thus, $\moins\tau$ is a minimal model of $\phi$ containing $x$ and $\card{\moins\tau}\leq k$,  hence $(\phi,k, x)$ is a positive instance of $\CMS^*(\rXOR{3})$.
\qed

% 4. XOR4 < XOR3, XOR2
\paragraph{\normalfont 4. $\CMS(\rXOR{4}) \reduc \CMS(\rXOR{3},\,\rXOR{2})$.}\ \\

Observe that $\rXOR{4}(x_1, x_2, x_3, x_4)\equiv \exists y, z : \rXOR{3}(x_1, x_2, y)\wedge \rXOR{3}(x_3, x_4, z)\wedge \rXOR{2}(y,z)$. The two fresh variables $y$ and $z$ take  complementary values, so they will together contribute a weight 1 in any case. 

% 5. XOR3, XOR2 < XOR3
\paragraph{\normalfont 5. $\CMS(\rXOR{3},\,\rXOR{2}) \reduc \CMS(\rXOR{3})$.}\ \\

Let $(\phi,x)$ be an instance of $\CMS(\rXOR{3},\,\rXOR{2})$. If $\phi$ is unsatisfiable, we map $(\phi, x)$ to a trivial negative instance of $\CMS(\rXOR{3})$, e.g.  $(\rXOR{3}(x_1,x_2,x_3), x)$.

Otherwise, we replace any constraint $\rXOR{2}(x,y)$ by $\rXOR{3}(x,y,w)$ where $w$ is a fresh variable of weight impact $N$ big enough, say bigger than the number of variables of the original formula. This assures that the cardinality-minimal models of the formula are the models of $\phi$ extended with $w=0$. The variable $w$ can be given the needed weight impact by adding the constraints $\rXOR{3}(t,t,t)\wedge 
\bigwedge_{i=1}^N \rXOR{3}(t,w,w_i)$ where $t$  and the $w_i$'s are fresh variables. 
\end{proof}

% section
\section{Proof of the main theorem}

We prove here Theorem \ref{theorem:classification}. The classification can be visualized on Post's Lattice, see figure~\ref{fig:posts-lattice}.
The classification obeys the borders among co-clones and, as discussed in the previous section, will be obtained by a systematic exploration of the co-clones.

Observe that Theorem \ref{theorem:krom_classification}, the previously obtained classification in the Krom fragment, concerns  co-clones in the lower part of the lattice, namely every  co-clone $\cocl$ such that $\cocl\subseteq\cc{ID}{}{2}$. In the depiction of Post's Lattice in figure~\ref{fig:posts-lattice} the color coding is as follows.
The ``white'' co-clones, for which the problem $\CMS$ is trivial, are the co-clones that contain only 0-valid relations. For those ones the cardinality-minimum solution is the all-0 solution, and the answer is always ``no''. 
The ``grey'' co-clones, for which the problem $\CMS$ is decidable in polynomial time,    correspond to co-clones $\cocl$ such that either $\cocl\subseteq\cc{IE}{}{2}$ or $\cocl\subseteq\cc{ID}{}{1}$. In the first case, all relations are Horn, and therefore  there exists a unique cardinality-minimal model that can be found by unit propagation in polynomial time. In the second case,  all relations are  width-2-affine and  the tractability result follows from Theorem \ref{theorem:krom_classification}.

Finally, to obtain the complexity classification it remains to prove hardness for the ''black'' co-clones, namely $\cc{II}{}{2}$, $\cc{II}{}{1}$, $\cc{IN}{}{2}$, $\cc{IL}{}{2}$, $\cc{IL}{}{3}$, $\cc{IL}{}{1}$, $\cc{IV}{}{2}$, $\cc{IV}{}{1}$,
%
%$\cc{IS}{}{00}$,  $\cc{IS}{}{01}$ ,  $\cc{IS}{}{02}$ ,  $\cc{IS}{}{0}$,
%
 and, for any $k\ge 2$, for the co-clones $\cc{IS}{k}{00}$, $\cc{IS}{k}{01}$, $\cc{IS}{k}{02}$, $\cc{IS}{k}{0}$. The ''black'' co-clone $\cc{ID}{}{2}$ is dealt with by Theorem~\ref{theorem:krom_classification}.

As we have discussed in the previous section, for each remaining co-clone $\cocl$, given one of its weak bases $B$ we will show that $\CMS(B)$ is hard. This will be done by a reduction from a known hard problem, either $\CMS(\rOR{2})$ or $\CMS(\rXOR{3})$. For example, given an instance $(\phi, x)$ of  $\CMS(\rOR{2})$, where $\phi$ is a conjunction of $\rOR{2}$-clauses, we will build a $B$-formula $\phi'$ such that $x$ belongs to a cardinality-minimal model of $\phi$ if and only if $x$ belongs to a cardinality-minimal model of $\phi'$. The  construction of $\phi'$ is obtained by a local replacement of each clause of $\phi$ by an equivalent $B$-formula. Usually this requires introduction of fresh (existentially quantified) variables. Some of these additional variables will be \emph{frozen}, which means that their truth value is the same in all models, and thus their contribution to the weight of any model is fixed. In order to be sure that the weight of the non-frozen additional variables will not 
compromise the cardinality-minimal models, the trick is to neutralize them by adding for each such variable $y$, another one $y'$ and to force them to take complementary values, i.e. $y\ne y'$. Thus the weight contribution of $y$ and $y'$ together will always be 1 in all models. Sometimes, to do so we will have to express the truth value 0. When this is not possible directly, the idea is to replace $0$ by a variable $f$, and then introduce a big number of copies of $f$ such that any cardinality-minimal model of the formula has to set $f$ to $0$.

In the following when we speak about \emph{the} minimal weak base of a co-clone we mean the weak base from Lagerkvist \cite{Lagerkvist14}, given in table~\ref{table:weak_bases}. In the proofs, we will always restate the exact definition of the corresponding weak base, and, where convenient, also its matrix representation.

The following proposition provides the missing hardness results.

\begin{proposition} \label{prop:all}
Let $\Gamma$ be a constraint language. Then $\CMS(\Gamma)$ is $\THETA{2}$-hard 
if $\clos{\Gamma} \in \{\cc{II}{}{2}, \cc{II}{}{1}, \cc{IN}{}{2}, \cc{IL}{}{2}, \cc{IL}{}{3}, \cc{IL}{}{1}, \cc{IV}{}{2}, \cc{IV}{}{1}, \cc{IS}{k}{00}, \cc{IS}{k}{01}, \cc{IS}{k}{02}, \cc{IS}{k}{0}\}$, for any $k \geq 2$.
\end{proposition}

The following lemmas provide the proof for Proposition~\ref{prop:all}, dealing with the different cases.

\begin{lemma} \label{prop:II2}
Let $\clos{\Gamma} = \cc{II}{}{2}$. Then $\CMS(\Gamma)$ is $\THETA{2}$-hard.
\end{lemma}

\begin{proof}
 Let $R_{\cc{II}{}{2}}$ be the minimal weak base of $\cc{II}{}{2}$, that is,
$\Rddd(x_1,\ldots,x_6) \wedge \F(x_7) \wedge \T(x_8)$, where 
$\Rddd=\{100011, 010101, 001110\}$. 
The matrix representation is as follows. 
%
% I2
\begin{center}
\mbox{
\begin{array}[t]{c}
R_{\cc{II}{}{2}} =  
\left(\begin{array}{c} 10001101 \\ 01010101 \\00111001 \end{array}\right)
\end{array}
}
\end{center}
We  show that $\CMS(\rOR{2}) \reduc \CMS(R_{\cc{II}{}{2}})$. Then the result follows from  Lemma \ref{lemma:cardminsat_completeness_or} and Proposition \ref{thm:weak_base_method}.

Let $(\phi, x)$ be an instance  of $\CMS(\rOR{2})$, where $\phi=\bigwedge_{i=1}^p (x_i^1\lor x_i^2)$. Let $\{a_i, b_i,  c_i, d_i, 
a'_i,  b'_i, c'_i,  d'_i\mid i=1\ldots p\}\cup \{t, f\}$ be fresh variables.
For each constraint $(x_i^1\lor x_i^2)$ we build the constraint 
$R_{\cc{II}{}{2}}(a_i, b_i, c_i, d_i, x_i^1, x_i^2,f, t)$.
Observe that $\rOR{2}(x_i^1,x_i^2 )\equiv$
$$\exists   a_i,b_i,c_i,d_i,f, t\  R_{\cc{II}{}{2}}(a_i, b_i, c_i, d_i,  x_i^1,x_i^2, f, t).$$
The existential variables are uniquely determined. The variables $f$ and $t$ are frozen, while the values of $a_i, b_i, c_i, d_i$  are not. Nevertheless their values can be neutralized by the introduction of additional fresh variables $a_i', b_i', c_i', d_i'$ who are forced to take complementary values. In the case of $a_i$ and $a_i'$ this can be achieved by the constraint
$R_{\cc{II}{}{2}}(a_i, a'_i, f, a_i', a_i, t, f, t)$.
%$R_{\cc{II}{}{2}}(f, a_i, a'_i, t, f, a_i, a'_i, t)$.
Analogous constraints are added for $b_i, b_i'$, $c_i, c_i'$ and $d_i, d_i'$. 
 
Consider $\phi '$ the conjunction of all these constraints.
Observe that the formulas $\phi$ and $\phi'$  are equivalent when quantifying on the fresh variables. Moreover, the models of $\phi$ and $\phi'$ are in one-to-one correspondence. Each model $\sigma$ of $\phi$ can be extended to a model $\sigma'$ of  $\phi'$ whose weight is $\vert \sigma'\vert = \vert \sigma \vert +4p+1$. 
Consequently, $x$ belongs to a cardinality-minimal model of $\phi$ if and only if $x$ belongs to a cardinality-minimal model of $\phi'$, thus concluding the proof.

\end{proof}

% proposition
\begin{lemma}
Let $\clos{\Gamma} =\cc{II}{}{1}$. Then 
$\CMS(\Gamma)$ is $\THETA{2}$-hard.
\end{lemma}
\begin{proof}
 Let $R_{\cc{II}{}{1}}$ be the minimal weak base of $\cc{II}{}{1}$, that is, $R_{\cc{II}{}{1}}(x_1, x_2, x_3, x_4)=(x_1 \vee x_2) \wedge (x_1x_2 \leftrightarrow x_3) \wedge \T(x_4)$. 
The matrix representation is as follows.
\begin{center}
\mbox{
\begin{array}[t]{c}
R_{\cc{II}{}{1}}=
\left(\begin{array}{c} 0101 \\ 1001 \\ 1111 \end{array}\right)  
\end{array}
}
\end{center}
We  show that $\CMS(\rOR{2}) \reduc \CMS(R_{\cc{II}{}{1}})$. Then the result follows from  Lemma \ref{lemma:cardminsat_completeness_or} and Proposition \ref{thm:weak_base_method}.

Let $(\phi, x)$ be an instance  of $\CMS(\rOR{2})$, where $\phi=\bigwedge_{i=1}^p (x_i^1\lor x_i^2)$.
%Let $\{y_i, z_i \mid i=1\ldots p\}\cup \{f^1_j, f^2_j \mid j=1\ldots N\} \cup\{t, f\}$ be fresh variables.
For each constraint $(x_i^1\lor x_i^2)$ we build the constraint 
$R_{\cc{II}{}{1}}( x_i^1, x_i^2, y_i, t)$. Observe that 
$(x_i^1\lor x_i^2)\equiv \exists y_i\exists  t R_{\cc{II}{}{1}}( x_i^1, x_i^2, y_i, t)$. The variable $t$ is frozen to 1. The variable $y_i$ is not, but we can  neutralize its weight by adding the constraint $R_{\cc{II}{}{1}}(y_i, z_i, f, t)$, which will force $z_i\equiv\neg y_i$ as soon as $f$ is evaluated to 0.
We force $f$ to be evaluated to 0 in any cardinality-minimal model by adding the constraints $R_{\cc{II}{}{1}}(f^1_j, f^2_j, f, t)$, for $j=1,\ldots, N$.
If $f=1$, these constraints force all the $f^1_j, f^2_j$ to 1, that is, $1+2N$ variables. If $f=0$, one of the $f^1_j, f^2_j$ is forced to 1 and the other to 0, that is, the weight contribution is only $N$.

Consider $\phi '$ the conjunction of all these constraints.
Observe that the formulas $\phi$ and $\phi'$ are equivalent when quantifying on the fresh variables.  Moreover, the models of $\phi$ and  the models of $\phi'$ in which $f=0$ are in one-to-one correspondence. Each model $\sigma$ of $\phi$ can be extended to a model $\sigma'$ of  $\phi'$ with $\sigma'(f)=0$,  whose weight is $\vert \sigma'\vert = \vert \sigma \vert +p+N+1$.

Observe that $\phi$ is always satisfiable and therefore, by the above observation, $\phi'$ always admits a model with $f=0$.
Moreover, the models of $\phi'$ in which $f=1$ are of cardinality at least $2N+2$, while the models of $\phi'$ in which $f=0$ are   of cardinality at most $n+p+N+1$, where $n$ is the number of variables of $\phi$.
Now, if we choose $N$ big enough, e.g. $N\ge p+n$, we ensure that an assignment with $f=1$ can never be a cardinality-minimal model. Consequently, putting all together shows that  $x$ belongs to a cardinality-minimal model of $\phi$ if and only if $x$ belongs to a cardinality-minimal model of $\phi'$, thus concluding the proof.
\end{proof}

% proposition
\begin{lemma}
Let $\clos{\Gamma} =\cc{IN}{}{2}$. Then 
$\CMS(\Gamma)$ is $\THETA{2}$-hard.
\end{lemma}

\begin{proof}
Let  $R_{\cc{IN}{}{2}}$ be the  minimal weak base of $ \cc{IN}{}{2}$, that is, 
 $R_{\cc{IN}{}{2}}=\evenn{4}{4 \neq}(x_1,\ldots,x_8) \wedge x_1x_4
\leftrightarrow x_2x_3$. The matrix representation is as follows.

% N2
\begin{center}
\mbox{
\begin{array}[t]{c}
R_{\cc{IN}{}{2}} = 
\left(\begin{array}{c} 00001111 \\ 00110011 \\ 01010101 \\ 10101010 \\ 11001100 \\ 11110000 \end{array}\right)
\end{array}
}
\end{center}

We  show that $\CMS(\rOR{2}) \reduc \CMS(R_{\cc{IN}{}{2}})$. Then the result follows from  Lemma \ref{lemma:cardminsat_completeness_or} and Proposition \ref{thm:weak_base_method}.
Observe that $\rOR{2}(x_i^1,x_i^2 )\equiv$
$$\exists   a_i,b_i,c_i,d_i \  R_{\cc{IN}{}{2}}(0, a_i, b_i, c_i, d_i,  x_i^1,x_i^2, 1).$$
In this co-clone we can express $f\ne t$, but not $f=0$ and $t=1$. The idea is to use $f$ and $t$ in place of 0 and 1 as global variables (that is, the same for all constraints) and to force them to take the appropriate values in all cardinality-minimal models. This can be done by adding the constraint $R_{\cc{IN}{}{2}}(f,f,f,f,t,t,t,t)$, which expresses $f\ne t$,  and by adding a number $N$  big enough of constraints $R_{\cc{IN}{}{2}}(f_j,f_j,f_j,f_j,t,t,t,t)$, where the $f_j$, for  $j=1,\ldots , N$, are fresh variables. 

In choosing $N$ bigger than the number of variables of the original formula,
we can assure that in any cardinality-minimal model $f$ is assigned 0 and $t$ is assigned 1. In using this trick we can mimic the reduction proposed in the proof of Proposition  \ref{prop:II2} and  hence transform an $\rOR{2}$-formula into an $R_{\cc{IN}{}{2}}$-formula in preserving the cardinality-minimal models, thus providing the reduction from $\CMS(\rOR{2})$ to $\CMS(R_{\cc{IN}{}{2}})$. 
\end{proof}

% proposition
% proposition
\begin{lemma}
Let $\clos{\Gamma} = \cc{IL}{}{2}$. Then $\CMS(\Gamma)$ is $\THETA{2}$-hard.
\end{lemma}

\begin{proof}
 Let $R_{\cc{IL}{}{2}}$ be the minimal weak base of $\cc{IL}{}{2}$, that is, $R_{\cc{IL}{}{2}} = \evenn{3}{3 \neq}(x_1,\ldots,x_6) \wedge
   \F(x_7) \wedge \T(x_8)$. The matrix representation is as follows.
%
% L2
\begin{center}
\mbox{
\begin{array}[t]{c}
R_{\cc{IL}{}{2}} =
\left(\begin{array}{c} 00011101 \\ 01110001 \\ 10101001 \\ 11000101\end{array}\right)  
\end{array}
}
\end{center}
We  show that $\CMS(\rXOR{3}) \reduc \CMS(R_{\cc{IL}{}{2}})$. Then the result follows from  Lemma \ref{lemma:cardminsat_completeness_xor} and Proposition \ref{thm:weak_base_method}.

Let $(\phi, x)$ be an instance  of $\CMS(\rXOR{3})$, where $\phi=\bigwedge_{i=1}^p (x_i^1\oplus x_i^2\oplus x_i^3)$.
Let $\{u_i, v_i,  w_i,  
u'_i,  v'_i, w'_i \mid i=1\ldots p\}\cup \{t, f\}$ be fresh variables.
For each constraint $(x_i^1\oplus x_i^2\oplus x_i^3)$ we build the constraint 
 $R_{\cc{IL}{}{2}}(u_i, v_i, w_i, x_i^1, x_i^2,x_i^3, f, t)$.
 Observe that $\rXOR{3}(x_i^1, x_i^2,x_i^3) \equiv$
$$\exists  f,t, u_i, v_i, w_i \ 
 R_{\cc{IL}{}{2}}(u_i, v_i, w_i, x_i^1, x_i^2,x_i^3, f, t).$$
 The variables   $f,t$ are frozen, and the   other existential  variables are uniquely determined and can be neutralized by adding three additional variables $u'_i, v'_i, w'_i $ and the constraint
$R_{\cc{IL}{}{2}}( u_i, v_i, w_i, u'_i, v'_i, w'_i, f, t)$.
 
Consider $\phi '$ the conjunction of all these constraints.
Observe that the formulas $\phi$ and $\phi'$  are equivalent when quantifying on the fresh variables. Moreover, the models of $\phi$ and $\phi'$ are in one-to-one correspondence. Each model $\sigma$ of $\phi$ can be extended to a model $\sigma'$ of  $\phi'$ whose weight is $\vert \sigma'\vert = \vert \sigma \vert +3p+1$. 
Consequently, $x$ belongs to a cardinality-minimal model of $\phi$ if and only if $x$ belongs to a cardinality-minimal model of $\phi'$, thus concluding the proof.
 \end{proof}

 % proposition
% proposition
\begin{lemma}
Let $\clos{\Gamma} = \cc{IL}{}{3}$. Then $\CMS(\Gamma)$ is $\THETA{2}$-hard.
\end{lemma}

\begin{proof}

Let $R_{\cc{IL}{}{3}}$ be the minimal weak base of $\cc{IL}{}{3}$, that is, $R_{\cc{IL}{}{3}} = \evenn{4}{4 \neq}(x_1,\ldots,x_8)$.
The matrix representation is as follows.
% 
% 
% % L3
% 
\begin{center}
\mbox{
 \begin{array}[t]{c}
 R_{\cc{IL}{}{3}} =
 
 \left(
 \begin{array}{c}
   00001111 \\  11000011 \\  10100101 \\  10010110 \\
  01101001 \\  01011010 \\  00111100 \\  11110000
 \end{array}
 \right)  
 \end{array}
}
\end{center}  
 
%\nadia{it seems to me that by now we can be faster, so I propose to keep Johannes idea, this is what follows}
 
We  show that $\CMS(\rXOR{3}) \reduc \CMS(R_{\cc{IL}{}{3}})$. Then the result follows from  Lemma \ref{lemma:cardminsat_completeness_xor} and Proposition \ref{thm:weak_base_method}.

All relations in the co-clone \cc{IL}{}{3} are complementive. In particular, 
 up to some permutation of the variables (columns) $R_{\cc{IL}{}{3}}$ is the complementive closure of $R_{\cc{IL}{}{2}}$.
Thus, we use the same reduction idea as for $\cc{IL}{}{2}$, and then 
 proceed analogously as we have done for the case of $\cc{IN}{}{2}$: replace 0 and 1 by $f$ and $t$, express $f\ne t \equiv R_{\cc{IL}{}{3}}(f,f,f,f,t,t,t,t)$ and put a big weight on $f$, hence forcing $f=0$ and $t=1$ in any cardinality-minimal model.

\end{proof}

\begin{lemma}
Let $\clos{\Gamma} = \cc{IL}{}{1}$. Then $\CMS(\Gamma)$ is $\THETA{2}$-hard.
\end{lemma}

\begin{proof}
 Let $R_{\cc{IL}{}{1}}$ be the minimal weak base of $\cc{IL}{}{1}$, that is, 
 $R_{\cc{IL}{}{1}}(x_1, x_2, x_3, x_4)=\rXOR{3}(x_1, x_2, x_3)  \wedge \T(x_4)$.
 The matrix representation is as follows.
 
%V1
\begin{center}
\mbox{
\begin{array}[t]{c}
R_{\cc{IL}{}{1}}=
\left(\begin{array}{c} 1001 \\ 0101 \\ 0011 \\ 1111  \end{array}\right)  
\end{array}
}
\end{center}

It is easy to  verify that 
$\rXOR{3}(x_1, x_2, x_3) \equiv  \exists\  t\;R_{\cc{IL}{}{1}}(x_1, x_2, x_3, t)$. Observe that the value of $t$ is  frozen to 1 in $R_{\cc{IL}{}{1}}(x_1, x_2, x_3, t)$.  
Thus, by the  introduction of a fresh variable $t$  and the local replacement of each clause   we can build a reduction from 
$\CMS(\rXOR{3}) $ to $\CMS(R_{\cc{IL}{}{1}})$. Then the result 
follows from  Lemma~6 and Proposition~4.
\end{proof}

% proposition
\begin{lemma}
Let $\clos{\Gamma} = \cc{IV}{}{2}$. Then 
$\CMS(\Gamma)$ is $\THETA{2}$-hard.
\end{lemma}
\begin{proof}
Let $R_{\cc{IV}{}{2}}$ be the minimal weak base of $\cc{IV}{}{2}$, that is,
 $R_{\cc{IV}{}{2}}(x_1, x_2, x_3, x_4,x_5)=(\overbar{x_1} \leftrightarrow
\overbar{x_2}\overbar{x_3}) \wedge \F(x_4) \wedge \T(x_5)$. 
The matrix representation is as follows.

% %V2
\begin{center}
\mbox{
\begin{array}[t]{c}
R_{\cc{IV}{}{2}} =
\left(\begin{array}{c} 00001 \\ 10101 \\ 11001 \\ 11101 \end{array}\right)  
\end{array}
}
\end{center}

It is easy to  verify that 
$\rOR{2}(x,y) \equiv  \exists \  t,f \;R_{\cc{IV}{}{2}}(t,x,y, f, t)$.
Observe that the values of $t$ and $f$  are frozen in $R_{\cc{IV}{}{2}}(t,x,y, f, t)$, they take respectively the values 1 and 0 in all models.  
Thus, we can reduce 
$\CMS(\rOR{2})$ to $\CMS(R_{\cc{IV}{}{2}})$. Then the result 
follows from Lemma~5 and Proposition~4.
\end{proof}

% proposition
\begin{lemma}
Let $\clos{\Gamma} = \cc{IV}{}{1}$. Then 
$\CMS(\Gamma)$ is $\THETA{2}$-hard.
\end{lemma}
\begin{proof}

Let $R_{\cc{IV}{}{1}}$ be the minimal weak base of $\cc{IV}{}{1}$, that is,
 $R_{\cc{IV}{}{1}}(x_1, x_2, x_3, x_4,x_5)=(\overbar{x_1} \leftrightarrow
\overbar{x_2}\overbar{x_3})\wedge (\overbar{x_2}\lor \overbar{x_3}\rightarrow \overbar{x_4})  \wedge \T(x_5)$. 
The matrix representation is as follows.
 
%V1
\begin{center}
\mbox{
\begin{array}[t]{c}
R_{\cc{IV}{}{1}}=
\left(\begin{array}{c} 00001 \\ 10101 \\ 11001 \\ 11101 \\ 11111 \end{array}\right)  
\end{array}
}
\end{center}
We prove that $\CMS(\rOR{2})\le \CMS(R_{\cc{IV}{}{1}})$.
Observe that 
$\rOR{2}(x,y) \equiv  \exists \ f,  t  \;R_{\cc{IV}{}{1}}(t,x,y, f, t)$.
Observe that the value  of $t$ is frozen to 1, but the value of  $f$ is not frozen. Nevertheless, any model with $f = 1$ will remain a model when flipping $f$ to 0. Therefore, this additional variable will always be set to 0 in a cardinality-minimal model.
Thus, by the  introduction of two fresh variables $f$ and  $t$  and the local replacement of each clause   we can build a reduction from 
$\CMS(\rOR{2}) $ to $\CMS(R_{\cc{IV}{}{1}})$. Then the result 
follows from  Lemma~5 and Proposition~4.
\end{proof}

 % proposition
% proposition
\begin{lemma}
Let $\clos{\Gamma} = \cc{IS}{k}{00}$, for some integer $k \geq 2$. Then $\CMS(\Gamma)$ is $\THETA{2}$-hard.
\end{lemma}
\begin{proof}

Let $R_\cc{IS}{k}{00}$ be the minimal weak base of $\cc{IS}{k}{00}$, that is, $R_\cc{IS}{k}{00}(x_1, \ldots, x_k, x_{k+1}, x_{k+2}, x_{k+3})= (x_1\lor \ldots \lor x_k)\land (x_{k+1}\rightarrow x_1\cdots x_k)\land \F(x_{k+2})\land \T(x_{k+3})$.

We prove that 
$\CMS(\rOR{2}) \le \CMS(R_\cc{IS}{k}{00})$. Then the result follows from  Lemma~5 and Proposition~4.
Let $\displaystyle (\phi=\bigwedge_{i=1}^p (x_i^1\lor x_i^2), x)$ be an instance  of $\CMS(\rOR{2})$. 
Let  $f$ and $t$ be  fresh variables. Consider the formula 
$\displaystyle\phi'=\bigwedge_{i=1}^p 
R_\cc{IS}{k}{00}(x_i^1, x_i^2, \ldots, x_i^2, f, f, t)$.
Observe that the variables $f$ and $t$ are frozen in $\phi'$. 
Hence, $x$ belongs to a cardinality-minimal model of $\phi$ if and only if $x$ belongs to a cardinality-minimal model of $\phi'$.
\end{proof}

 % proposition
% proposition
\begin{lemma}
Let $\clos{\Gamma} = \cc{IS}{k}{01}$, for some integer $k \geq 2$. Then $\CMS(\Gamma)$ is $\THETA{2}$-hard.
\end{lemma}
\begin{proof}
Let $R_\cc{IS}{k}{01}$ be the minimal weak base of $\cc{IS}{k}{01}$, that is,

$R_\cc{IS}{k}{01}(x_1, \ldots, x_k, x_{k+1}, x_{k+2}) = (x_1\lor \ldots \lor x_k)\land (x_{k+1}\rightarrow x_1\cdots x_k)\land \T(x_{k+2})$.
We prove that 
$\CMS(\rOR{2}) \le \CMS(R_\cc{IS}{k}{01})$. Then the result follows from  Lemma~5 and Proposition~4.
Let $\displaystyle(\phi=\bigwedge_{i=1}^p (x_i^1\lor x_i^2), x)$ be an instance  of 
$\CMS(\rOR{2})$. 
Let  $f$ and $t$ be  fresh variables. Consider the formula 
$\displaystyle\phi'=\bigwedge_{i=1}^p 
R_\cc{IS}{k}{01}(x_i^1, x_i^2, \ldots, x_i^2, f,  t)$.
This time $t$ is frozen, but $f$ is not. Nevertheless, if there is any model with $f=1$, it will remain a model when flipping $f$ to 0. Therefore, this reduction preserves the cardinality-minimal models. 
Hence, $x$ belongs to a cardinality-minimal model of $\phi$ if and only if $x$ belongs to a cardinality-minimal model of $\phi'$.
\end{proof}

 In exactly the same way, but even more easily, we obtain the following.
 
  % proposition
% proposition
\begin{lemma}
Let $\clos{\Gamma} = \cc{IS}{k}{02}$ or $\clos\Gamma = \cc{IS}{k}{0}$. Then $\CMS(\Gamma)$ is $\THETA{2}$-hard.
\end{lemma}
\begin{proof}
We use exactly the same reduction as for $\cc{IS}{k}{01}$ and $\cc{IS}{k}{00}$, the only difference in the weak bases is the non-presence of the clause $(x_{k+1}\rightarrow x_1\cdots x_k)$ which will be omitted.
\end{proof}

This concludes the proof of Theorem~\ref{theorem:classification}.
We restate here Theorem~\ref{theorem:classification} in terms of co-clones.
 
% theorem
\begin{theorem}
Let $\Gamma$ be a finite constraint language. The problem $\CMS(\Gamma)$ is
\begin{itemize}
\item $\THETA{2}$-complete if $\cocl \subseteq \clos{\Gamma} \subseteq \cc{II}{}{2}$
for $\cocl \in \{\cc{IS}{2}{0}, \cc{IL}{}{3}, \cc{IL}{}{1}\}$,
%
%\item polynomial time solvable if   either $\clos{\Gamma} = \cc{ID}{}{}$ or  $\cc{IR}{}{1} \subseteq \clos{\Gamma} \subseteq \cocl$, for $\cocl \in \{\cc{ID}{}{1}, \cc{IE}{}{2}\}$, 
%
\item polynomial time solvable if either $\cc{ID}{}{} \subseteq \clos{\Gamma} \subseteq \cc{ID}{}{1}$ or  $\cc{IR}{}{1} \subseteq \clos{\Gamma} \subseteq \cc{IE}{}{2}$, 
\item trivial otherwise ($\Gamma$ is 0-valid)
\end{itemize}
\end{theorem}

%\begin{figure}
%\centering
%\includegraphics[width=17cm]{posts-lattice}
%\caption{Complexity overview for $\CMS$.}
%\label{fig:posts-lattice}
%\end{figure}

\section{Example of application}

Let us now consider the following relevance problem for abduction. A {\em propositional abduction problem\/} (PAP)~$\pap$ consists of a
tuple $\tuple{V,H,M,T}$, where~$V$ is a finite set of \emph{variables}, $H \subseteq V$ is the set of \emph{hypotheses}, $M\subseteq V$ is the set of \emph{manifestations}, and~$T$ is a consistent \emph{theory} in the form of a propositional formula. A set $\s \subseteq H$ is a \emph{solution} (also called \emph{explanation}) to~$\pap$ if $T \cup \s$ is consistent and $T \cup \s \models M$ holds. 
Often, one is not interested in any solution of a given PAP $\pap$ but only in minimal
solutions, where minimality is defined w.r.t. set inclusion or smaller cardinality.

For subset-minimality the relevance problem has been completely classified in Schaefer's framework by Creignou and Zanuttini \cite{siamcomp/CreignouZ06}. 
 Here we  consider the following decision problem.

\noindent
\dproblem{\relevance }%
{PAP $\pap = \tuple{V,H,M,T}$ and hypothesis $h \in H$.}%
{Is $h$ relevant, i.e., does $\pap$ admit a cardinality-minimal solution $\s$ such that $h \in \s$?}

\smallskip

It is known that the \relevance problem is $\THETA 3$-complete in its full generality and $\THETA 2$-complete in the Horn case \cite{jacm/EiterG95}. The Krom case has been considered afterwards \cite{CreignouPW18}.
  The complexity results obtained so far  for the \relevance problem  were  restricted, due to an incomplete picture of the complexity of $\Cardminsat$. 
With the help of Theorem \ref{theorem:classification} we extend these results in showing that the complexity of \relevance in the affine case matches the Horn and Krom cases.

% theorem
% theorem
\begin{theorem} \label{theorem:abduction}
 \relevance is $\THETA 2$-complete   even if the theory is restricted to XOR-CNF-formulas.
\end{theorem}
\begin{proof} 
Membership follows from the fact that one can decide the satisfiability of an XOR-CNF formula in polynomial time.
 The hardness proof is obtained via a reduction from $\CMS(\rXOR{3}).$

 Consider an
arbitrary instance $(\phi, x_i)$ of $\CMS(\rXOR{3})$. 
Let $\phi =  \bigwedge_{i=1}^p ( x_i^1\oplus x_i^2\oplus x_i^3)$
over variables $X = \{x_1, \dots, x_n\}$
and let $G = \{g_1, \dots,  g_p \}$ be a set of fresh, pairwise distinct variables.
We define the PAP $\pap = \tuple{V,H,M,T}$ as follows: 
  \begin{eqnarray*}
    V &=& X \cup G \\
    H & = & X \\
    M & = &  G\\
    T &=& \{(x_i^1\oplus x_i^2\oplus x_i^3\oplus \bar g_i) \mid 1 \leq i \leq p \}
  \end{eqnarray*}
It is easy to verify that 
the models of $\phi$ coincide with the solutions of $\pap$. Hence,
$x_i$ is in a cardinality-minimal model of $\phi$ if and only if 
$x_i$ is in a cardinality-minimal solution of $\pap$.

Note that, more precisely, the proof shows the hardness of \relevance for $\rxor{4}$-formulas.

\end{proof}

\section{Conclusion}

In this paper we obtained a complete complexity classification of the problem $\Cardminsat(\Gamma)$ for all finite constraint languages $\Gamma$: if $\Gamma$ is width-2-affine, Horn or 0-valid,  $\Cardminsat(\Gamma)$ is solvable in polynomial time, otherwise it is $\THETA 2$-complete. The weak base method developed by Schnoor and Schnoor \cite{scs08}, completed with  the description of  minimal weak bases for co-clones by Lagerkvist \cite{Lagerkvist14} proved to be a valuable tool for this endeavor. As described  in the introduction understanding the complexity of $\Cardminsat$ is crucial for the study of several reasoning tasks in artificial intelligence that are based on minimizing cardinality. As we have motivated and outlined above we believe the establishment of the complete complexity picture of $\Cardminsat(\Gamma)$ is a cornerstone for future research in this direction:
it will allow the precise analysis of the computational complexity of problems such as relevance questions and belief revision operators. To obtain a richer picture we further plan to investigate the parametrized complexity of such problems. For instance, in \cite{MaMeSc21} a rich picture of the parametrized complexity of abduction problems is obtained. Yet, the named abduction relevance problem in this picture is missing. With the now established complete complexity classification  of $\Cardminsat$ it seems in reach to complete this picture.

% section
\section*{Acknowledgments}
This research has been supported by the Centre International de Recherche Math\'ematique (CIRM), Project reference \textit{Research in pairs} 1886.

\bibliographystyle{plain}

\renewcommand{\arraystretch}{1.4}

\begin{table*}[t]  \scriptsize
\centering
\begin{tabularx}{\textwidth+.1cm}{l l l}
  \hline
  Co-clone  & Minimal weak base & Name/Indication \\
  \hline
  $\cc{IBF}{}{}$  & $(x_1 = x_2)$ & -\\
  $\cc{IR}{}{0}$  & $\F(x_1)$ & -\\ 
  $\cc{IR}{}{1}$  & $\T(x_1)$ & -\\ 
  $\cc{IR}{}{2}$  & $\F(x_1) \wedge \T(x_2)$ & -\\
  $\cc{IM}{}{}$  & $(x_1 \rightarrow x_2)$ & implicative and 0- and 1-valid\\  
  $\cc{IM}{}{0}$  & $(x_1 \rightarrow x_2) \wedge \F(x_3)$ & implicative and 0-valid\\
  $\cc{IM}{}{1}$  & $(x_1 \rightarrow x_2) \wedge \T(x_3)$ & implicative and 1-valid \\ 
  $\cc{IM}{}{2}$ & $(x_1 \rightarrow x_2) \wedge \F(x_3) \wedge \T(x_4)$ & implicative \\
  $\cc{IS}{k}{0}, k \geq 2 $ & $\orn{}{k}(x_1, \ldots, x_{k})   \wedge \T(x_{k+1})$ & positive of width $k$\\

  $\cc{IS}{k}{02}, k \geq 2$ & $\orn{}{k}(x_1, \ldots, x_{k}) \wedge
  \F(x_{k+1}) \wedge \T(x_{k+2})$ & essentially positive of width $k$ \\

  $\cc{IS}{k}{01}, k \geq 2$ & $\orn{}{k}(x_1, \ldots, x_{k}) \wedge (x_{k+1} \rightarrow x_1 \cdots x_{k}) \wedge \T(x_{k+2})$ & - \\

  $\cc{IS}{k}{00}, k \geq 2$ &  $\orn{}{k}(x_1, \ldots, x_{k})   \wedge (x_{k+1} \rightarrow x_1 \cdots x_{k}) \wedge \F(x_{k+2}) \wedge \T(x_{k+3})$ & IHS-B+ of width $k$\\

  $\cc{IS}{k}{1}, k \geq 2$ & $\nandn{}{k}(x_1, \ldots, x_{k}) \wedge   \F(x_{k+1})$ & negative of width $k$ \\

  $\cc{IS}{k}{12}, k \geq 2 $ & $\nandn{}{k}(x_1, \ldots, x_{k}) \wedge
  \F(x_{k+1}) \wedge \T(x_{k+2})$ & essentially negative of width $k$  \\

  $\cc{IS}{k}{11}, k \geq 2$ & $\nandn{}{k}(x_1, \ldots, x_{k})   \wedge (x_1 \rightarrow x_{k+1}) \land \ldots \land (x_{k} \rightarrow x_{k+1}) \land \F(x_{k+2})$ & -\\

  $\cc{IS}{k}{10}, k \geq 2$ & $\nandn{}{k}(x_1, \ldots, x_{k})
  \wedge (x_1 \rightarrow x_{k+1}) \land \ldots \land (x_{k} \rightarrow x_{k+1})
   \wedge \F(x_{k+2}) \wedge \T(x_{k+3})$ & IHS-B- of width $k$ \\

  $\cc{ID}{}{}$  & $(x_1 \neq x_2)$ & strict 2-affine\\

  $\cc{ID}{}{1}$  & $(x_1 \neq x_2) \wedge \F(x_{3}) \wedge \T(x_{4})$ & 2-affine\\

  $\cc{ID}{}{2}$  & $\orn{}{2}(x_1,x_2) \land (x_1 \neq x_3) \land (x_2 \neq x_4) \land \F(x_{5}) \wedge \T(x_6)$ & Krom, bijunctive, 2CNF\\

  $\cc{IL}{}{}$   & $\evenn{4}{}(x_1,x_2,x_3,x_4)$ & affine and 0- and 1-valid\\ 

  $\cc{IL}{}{0}$  & $\evenn{3}{}(x_1,x_2,x_3) \wedge \F(x_4)$ & affine and 0-valid \\
	
  $\cc{IL}{}{1}$  & $\rXOR{3}{}(x_1,x_2,x_3) \wedge \T(x_4)$ & affine and 1-valid\\ 
	
  $\cc{IL}{}{2}$ & $\evenn{3}{3 \neq}(x_1,\ldots,x_6) \wedge \F(x_7) \wedge \T(x_8)$  & affine\\
	
  $\cc{IL}{}{3}$ & $\evenn{4}{4 \neq}(x_1,\ldots,x_8)$ &  -\\

  $\cc{IV}{}{}$  & $(\overbar{x_1} \leftrightarrow \overbar{x_2}\overbar{x_3}) \wedge (\overbar{x_2} \vee \overbar{x_3} \rightarrow \overbar{x_4})$ & dualHorn and 1- and 0-valid \\

  $\cc{IV}{}{0}$  & $(\overbar{x_1} \leftrightarrow
\overbar{x_2}\overbar{x_3}) \wedge \F(x_4)$  & definite dualHorn \\

  $\cc{IV}{}{1}$ & $(\overbar{x_1} \leftrightarrow \overbar{x_2}\overbar{x_3}) \wedge (\overbar{x_2} \vee \overbar{x_3} \rightarrow \overbar{x_4}) \wedge \T(x_5)$  & dualHorn and 1-valid\\

  $\cc{IV}{}{2}$  & $(\overbar{x_1} \leftrightarrow \overbar{x_2}\overbar{x_3}) \wedge \F(x_4) \wedge \T(x_5)$  & dualHorn\\
	
  $\cc{IE}{}{}$  & $(x_1 \leftrightarrow x_2x_3) \wedge (x_2 \vee x_3 \rightarrow x_4)$ & Horn and 0- and 1-valid\\ 
	
  $\cc{IE}{}{0}$  & $(x_1 \leftrightarrow x_2x_3) \wedge (x_2 \vee x_3 \rightarrow x_4) \wedge \F(x_5)$ & Horn and 0-valid\\
	
  $\cc{IE}{}{1}$  & $(x_1 \leftrightarrow x_2x_3) \wedge \T(x_4)$ & definite Horn\\
	
  $\cc{IE}{}{2}$ & $(x_1 \leftrightarrow x_2x_3) \wedge \F(x_4) \wedge \T(x_5)$  & Horn \\

  $\cc{IN}{}{}$  & $\evenn{4}{}(x_1,x_2,x_3,x_4) \wedge x_1x_4 \leftrightarrow x_2x_3$ & complementive and 0- and 1-valid\\
	
  $\cc{IN}{}{2}$ & $\evenn{4}{4 \neq}(x_1,\ldots,x_8) \wedge x_1x_4 \leftrightarrow x_2x_3$  & complementive\\

  $\cc{II}{}{}$  & $(x_1 \leftrightarrow x_2x_3) \wedge (\overbar{x_4}
\leftrightarrow \overbar{x_2}\overbar{x_3})$ & 0- and 1-valid\\

  $\cc{II}{}{0}$  & $(\overbar{x_1} \vee \overbar{x_2}) \wedge (\overbar{x_1}\overbar{x_2} \leftrightarrow \overbar{x_3}) \wedge \F(x_4)$  & 0-valid\\

  $\cc{II}{}{1}$  & $(x_1 \vee x_2) \wedge (x_1x_2 \leftrightarrow x_3) \wedge \T(x_4)$ & 1-valid\\

  $\cc{II}{}{2}$ & $\Rddd(x_1,\ldots,x_6) \wedge \F(x_7) \wedge \T(x_8)$ &  all Boolean relations  \\
  \hline
\end{tabularx}
\caption{All finitely generated co-clones with minimal weak bases from Lagerkvist \cite{Lagerkvist14}.}
\label{table:weak_bases}
\end{table*}

\begin{figure*}[t]
\centering
\includegraphics[width=0.99\textwidth]{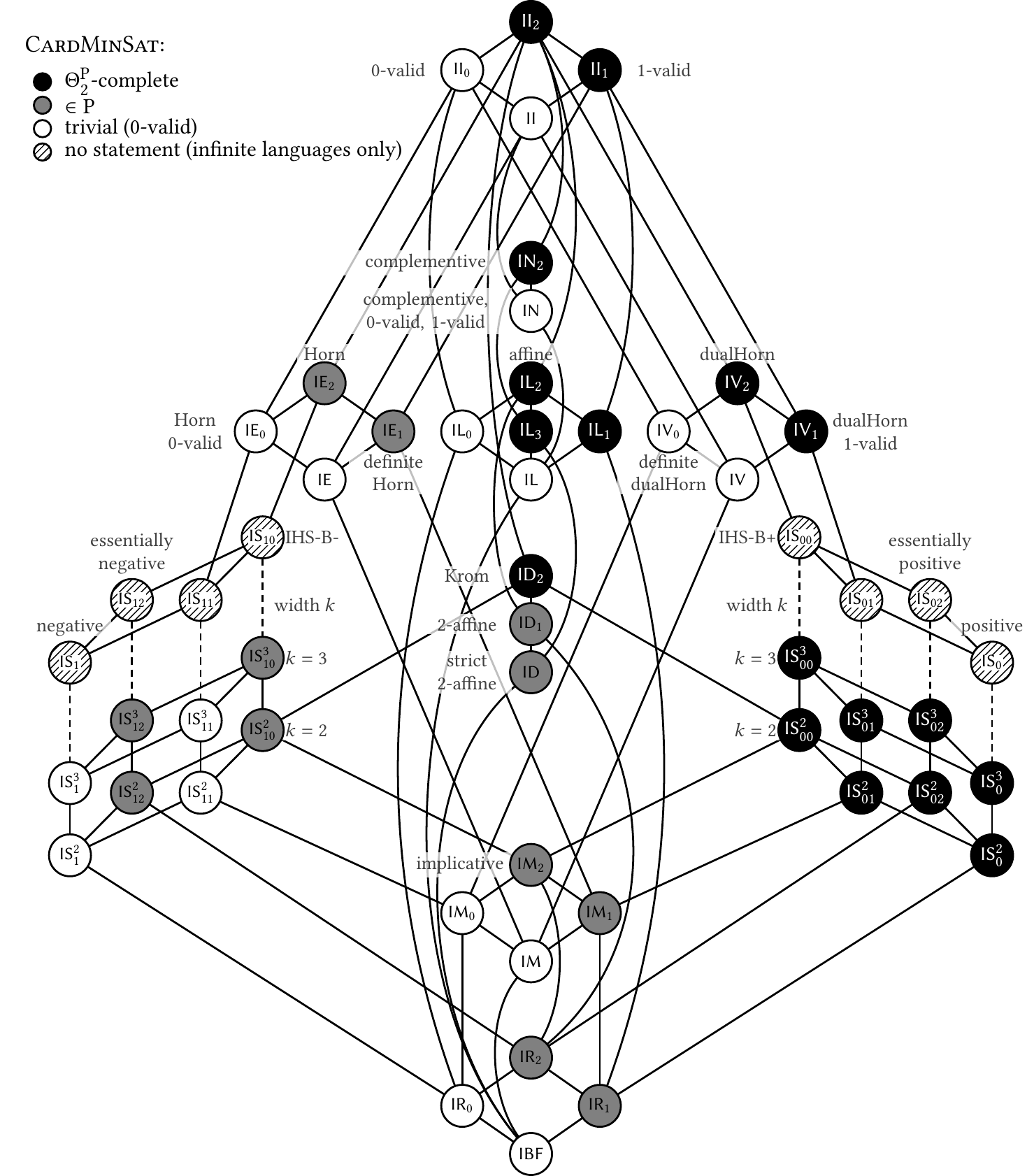} % Reduce the figure size so that it is slightly narrower than the column.
\caption{Complexity overview for $\Cardminsat$ illustrated on Post's Lattice.}
\label{fig:posts-lattice}
\end{figure*}

\end{document}